\documentclass[leqno, 11pt]{article}
\usepackage{hyperref}
\usepackage[left=2cm,right=2cm,top=1.4cm,bottom=1.1cm,includeheadfoot]{geometry}
\usepackage[utf8x]{inputenc}
\usepackage[english]{babel}
\usepackage{amssymb,amsmath,amstext,amsthm,amsfonts}

\usepackage{fancyhdr}
\newtheorem{thm}{Theorem}[section]

\newtheorem{prop}[thm]{Proposition}

\theoremstyle{definition}

\numberwithin{equation}{section}

\newcommand{\dom}{\mathrm{dom}}

\newcommand{\cl}{\mathcal}

\newcommand{\C}{\mathbb{C}}

\DeclareMathOperator*{\esssup}{ess\,sup}

\newcommand\gotX{{\mathfrak{X}}}

\newcommand\dR{{\mathbb{R}}}
\newcommand\dC{{\mathbb{C}}}

\newcommand\dN{{\mathbb{N}}}

\newcommand{\ga}{{\alpha}}
\newcommand{\gb}{{\beta}}
\newcommand{\gd}{{\delta}}
\newcommand{\gD}{{\Delta}}
\newcommand{\gga}{{\gamma}}

\newcommand{\go}{{\omega}}

\newcommand\gt{{\tau}}

\newcommand{\gT}{{\Theta}}

\newcommand\cA{{\mathcal{A}}}
\newcommand\cB{{\mathcal{B}}}
\newcommand\cC{{\mathcal{C}}}

\newcommand\cG{{\mathcal{G}}}

\newcommand\cI{{\mathcal{I}}}

\newcommand\cK{{\mathcal{K}}}

\newcommand\cO{{\mathcal{O}}}

\newcommand\cU{{\mathcal{U}}}

\newcommand{\ba}{\begin{array}}
\newcommand{\ea}{\end{array}}
\newcommand{\bea}{\begin{eqnarray}}
\newcommand{\eea}{\end{eqnarray}}
\newcommand{\bead}{\begin{eqnarray*}}
\newcommand{\eead}{\end{eqnarray*}}
\newcommand{\be}{\begin{equation}}
\newcommand{\ee}{\end{equation}}
\newcommand{\bed}{\begin{displaymath}}
\newcommand{\eed}{\end{displaymath}}
\newcommand{\bs}{\begin{split}}
\newcommand{\ees}{\end{split}}

\def\XXint#1#2#3{{\setbox0=\hbox{$#1{#2#3}{\int}$ }
\vcenter{\hbox{$#2#3$ }}\kern-.6\wd0}}

\title{Remarks on the operator-norm
convergence of the Trotter product formula}
\author{Hagen \textsc{Neidhardt}\footnote{H. Neidhardt: WIAS Berlin, Mohrenstr. 39, D-10117 Berlin, Germany;
email: hagen.neidhardt@wias-berlin.de}, 
Artur \textsc{Stephan}\footnote{A. Stephan: HU Berlin, Institut f\"ur Mathematik, Unter den Linden 6, D-10099 Berlin, Germany;
email: stephan@math.hu-berlin.de}, and Valentin A. \textsc{Zagrebnov}\footnote{V.A.Zagrebnov: Universit\'{e} d'Aix-Marseille - Institut de Math\'{e}matiques de Marseille  (UMR 7373), CMI - Technop\^{o}le Ch\^{a}teau-Gombert, 39, rue F. Joliot Curie, 13453 Marseille, France, email: valentin.zagrebnov@univ-amu.fr}}

\begin{document}
\maketitle

\begin{abstract}

We revise the operator-norm convergence of the Trotter product formula for a pair
$\{A,B\}$ of generators of semigroups on a Banach space.
Operator-norm convergence holds true if the dominating operator $A$ generates a
holomorphic contraction semigroup and $B$ is a $A$-infinitesimally small generator of a contraction semigroup, in particular, if $B$ is a bounded operator.
Inspired by studies of evolution semigroups
it is shown in the present paper that the operator-norm convergence
generally fails even for bounded operators $B$ if $A$ is not a
holomorphic generator. Moreover, 
it is shown that operator norm convergence of the Trotter product formula can be arbitrary slow.
\end{abstract}
\bigskip\bigskip\bigskip

%\newpage
%\tableofcontents 

%%%%%%%%%%%%%%%%%%%%%%%%%%%%%%%%%%%%%%%%%%%% Section %%%%%%%%%%%%%%%%%%%%%%%%%%%%%%%%%%%%%%%%%%%%%%%%%%%%%%%%
\section{Introduction and main results}\label{Intr}
%%%%%%%%%%%%%%%%%%%%%%%%%%%%%%%%%%%%%%%%%%%%%%%%%%%%%%%%%%%%%%%%%%%%%%%%%%%%%%%%%%%%%%%%%%%%%%%%%%%%%%%%%%%%%
Recall that the product formula
\begin{equation*}
e^{-\gt C} = \lim_{n\to\infty}\left(e^{-\gt A/n}e^{-\gt B/n}\right)^n, \quad \gt \ge 0,
\end{equation*}
was established by S.~Lie (in 1875) for matrices where $C := A + B$. The proof is based on the
telescopic representation
\be\label{eq:1.2}
\big(e^{-\gt A/n}e^{-\gt B/n}\big)^n - e^{-\gt C} 
= \sum^{n-1}_{k=0}\left(e^{-\gt A/n}e^{-\gt B/n}\right)^{n-1-k}\left(e^{-\gt A/n}e^{-\gt B/n} -
e^{-\gt C/n}\right)e^{-k\gt C/n} \,,
\ee
$n \in \dN$, and expansion
\begin{equation*}
e^{-\gt X}  = I - \gt X + O(\gt^2), \qquad \gt \longrightarrow 0,
\end{equation*}
for a matrix $X$ in the operator-norm topology $\|\cdot\|$. Indeed, using this expansion one
obtains the estimate:
\begin{equation*}
\|e^{-\gt A/n}e^{-\gt B/n} - e^{-\gt C/n}\| = O(({\gt}/{n})^2).
\end{equation*}
{{Then from \eqref{eq:1.2} we get the existence of a constant $c_0 > 0$ such that the
following estimate holds}}
\begin{equation*}
\big\|\big(e^{-\gt A/n}e^{-\gt B/n}\big)^n - e^{-\gt C}\|
\le c_0 \frac{\gt^2}{n^2}\sum^{n-1}_{k=0}e^{\gt\tfrac{n-1-k}{n}\gt \|A\|}e^{\gt\tfrac{n-1-k}{n}\gt \|B\|}
e^{\gt\tfrac{k}{n}\|C\|} \ .
\end{equation*}
Since $\|C\| \le \|A\| + \|B\|$, one obtains inequality
\begin{equation*}
\big\|\big(e^{-\gt A/n}e^{-\gt B/n}\big)^n - e^{-\gt C}\| 
\le c_0 \frac{\gt^2}{n^2}\sum^{n-1}_{k=0}e^{\gt\tfrac{n-1}{n} (\|A\| + \|B\|)} \le c_0 \frac{\gt^2}{n}
e^{\gt(\|A\| + \|B\|)} \ ,
\end{equation*}
which yields that
\be\label{eq:1.7}
\sup_{\gt \in [0,T]}\big\|\big(e^{-\gt A/n}e^{-\gt B/n}\big)^n - e^{-\gt C}\| = O({1}/{n}) \ ,
\ee
as $n \to \infty$ for any $T > 0$. Note that this proof carries through verbatim for bounded operators
$A$ and $B$ on Banach spaces.

H.~Trotter \cite{Trotter1959} has extended this result to unbounded operators $A$ and $B$ on Banach spaces,
but in the strong operator topology. He proved that if $A$ and $B$ are generators of contractions semigroups
on a separable Banach space such that the algebraic sum $A+B$ is a densely defined closable operator and the closure
$C = \overline{A+B}$ is a generator of a contraction semigroup, then
\begin{align}\label{eq:1.8}
e^{-\gt C} = \underset{n\rightarrow\infty}{\mathrm{s- lim}}~\big(e^{-\gt A/n}e^{-\gt B/n}\big)^n \ ,
\end{align}
uniformly in $\gt \in [0,T]$ for any $T > 0$.
%Notice that the topology has changed form operator-norm topology to the strong operator topology.
It is obvious that this result holds if $B$ is a bounded operator.

Considering the Trotter product formula on a Hilbert space T.~Kato has shown in \cite{Kato1978} that for
non-negative operators $A$ and $B$ the Trotter formula \eqref{eq:1.8} holds in the \textit{strong} operator
topology if $\dom(\sqrt{A}) \cap \dom(\sqrt{B})$ is dense in the Hilbert space and $C = A \dot+ B$
is the form-sum of operators $A$ and $B$. Later on it was shown in \cite{ITTZ2001} that the relation
\eqref{eq:1.7} holds if the algebraic sum $C = A+B$ is already a self-adjoint operator. Therefore,
\eqref{eq:1.7} is valid in particular, if $B$ is a bounded self-adjoint operator.

The historically first result concerning the operator-norm convergence of the Trotter formula in
a Banach space is due to \cite{CachZag2001}. Since the concept of self-adjointness is missing for Banach
spaces it was assumed that the \textit{dominating} operator $A$ is a generator of a \textit{contraction holomorphic}
semigroup and $B$ is a generator of a contraction semigroup. In Theorem 3.6 of \cite{CachZag2001} it was
shown that if $0 \in \rho(A)$ and if there is a $\ga \in [0,1)$ such that $\dom(A^\ga) \subseteq \dom(B)$
and $\dom(A^*) \subseteq \dom(B^*)$, {then for any $T > 0$ one has}
\be\label{eq:1.9}
\sup_{\gt \in [0,T]}\big\|\big(e^{-\gt A/n}e^{-\gt B/n}\big)^n - e^{-\gt C}\| = O({\ln(n)}/{n^{1-\ga}}) \ .
\ee

Note that the assumption $0 \in \rho(A)$ was made for simplicity and that
the assumption $\dom(A^\ga) \subseteq \dom(B)$ yields that the operator $B$ is infinitesimally small with
respect to $A$. Taking into account \cite[Corollary IX.2.5]{Kato1980} one gets that the well-defined
algebraic sum $C = A+B$ is a generator of a contraction holomorphic semigroup.
By Theorem 3.6 of \cite{CachZag2001} the convergence rate \eqref{eq:1.9} improves if $B$ is a bounded
operator, i.e. $\alpha =0$.
Then for any $T > 0$ one gets
\begin{equation*}
\sup_{\gt \in [0,T]}\big\|\big(e^{-\gt A/n}e^{-\gt B/n}\big)^n - e^{-\gt C}\| = O({(\ln(n))^2}/{n}) \ .
\end{equation*}

Summarizing, the question arises whether the Trotter product formula converges in the operator-norm
if $A$ is a generator of a contraction (but not holomorphic) semigroup and $B$ is a bounded operator?
The aim of the present paper is to give an answer to this question for a certain class of generators.

It turns out that an appropriate class for that  is the class of generators of \textit{evolution}
semigroups. To proceed further we need
the notion of a \textit{propagator}, or a \textit{solution operator} \cite{NeiStephanZag2016}.

A strongly continuous map $U(\cdot,\cdot): \gD \longrightarrow \cB(X)$,
where $\gD := \{(t,s): 0 < s \le t \le T\}$ and $\cB(X)$ is the set of bounded operators on the separable
Banach space $X$, is called a \textit{propagator} if the conditions
%
%
%\begin{enumerate}
\begin{eqnarray*}
&& {\mathrm{(i)}} \ \sup_{(t,s)\in\gD}\|U(t,s)\|_{\cB(X)} < \infty \ , \\
&& {\mathrm{(ii)}} \ U(t,s) = U(t,r)U(r,s), \ 0 < s \le r \le t \le T \ ,
\end{eqnarray*}
are satisfied. Let us consider the Banach space $L^p(\cI,X)$, $\cI := [0,T]$, $p \in [1,\infty)$.
The operator $\cK$ is an evolution generator of the evolution semigroup $\{e^{-\gt \cK}\}_{\tau\geq0}$ if
there is a propagator such that the representation
\begin{equation}\label{eq:1.5}
(e^{-\gt \cK}f)(t) = U(t,t-\gt)\chi_{\cI}(t-\gt)f(t-\gt), \quad f \in L^p(\cI,X) \ ,
\end{equation}
holds for a.e. $t \in \cI$ and $\gt \ge 0$ \cite{NeiStephanZag2016}.
Since $e^{-\gt \cK} f = 0$ for $\gt \ge T$, the evolution generator $\cK$ can never be a generator of
a holomorphic semigroup.

A simple example of an evolution generator is the differentiation operator:
\begin{equation}\label{eq:1.6}
\begin{split}
(D_0f)(t) &:= \partial_{t} f(t), \\
f \in \dom(D_0) &:= \{f \in H^{1,p}(\cI,X): f(0) = 0\}.
\end{split}
\end{equation}
Then by (\ref{eq:1.6}) one obviously gets the contraction shift semigroup:
\begin{equation}\label{eq:1.61}
(e^{-\gt D_0}f)(t) = \chi_{\cI}(t-\gt)f(t-\gt), \quad f \in L^p(\cI,X),
\end{equation}
for a.e. $t \in \cI$ and $\gt \ge 0$. Hence, (\ref{eq:1.5}) implies that the corresponding
propagator of the non-holomorphic evolution semigroup $\{e^{-\gt D_0}\}_{\tau\geq0}$ is given by
$U_{D_0}(t,s) = I$, $(t,s) \in \gD$.

Note that in \cite{NeiStephanZag2016} we considered the operator $\cK_0 := \overline{D_0 + \cA}$,
where $\cA$ is the multiplication operator induced by a generator $A$ of a holomorphic contraction semigroup
on $X$. More precisely
\begin{equation*}
\begin{split}
(\cA f)(t) &:= Af(t), \ {\mathrm{and}} \ (e^{-\gt \cA} f)(t) = e^{-\gt A}  f(t) \ , \\
f \in \dom(\cA) &:= \{f \in L^p(\cI,X): Af(\cdot) \in L^p(\cI,X)\} \ .
\end{split}
\end{equation*}
Then the perturbation of the shift semigroup (\ref{eq:1.61}) by $\cA$ corresponds to the semigroup with
generator $\cK_0$. One easily checks that $\cK_0$ is an evolution generator of a contraction semigroup
on $L^p(\cI,X)$ that is never holomorphic. Indeed, since the generators $D_0$ and $\cA$ commute, the
representation (\ref{eq:1.5}) for evolution semigroup $\{e^{-\gt \cK_0}\}_{\tau\geq0}$ takes the form:
\begin{equation*}
(e^{-\gt \cl K_0}f)(t) = e^{-\gt A}\chi_{\cI}(t-\gt)f(t-\gt), \quad f \in L^p(\cI,X) \ ,
\end{equation*}
for a.e. $t \in \cI$ and $\gt \ge 0$ with propagator $U_{0}(t,s) = e^{-(t-s) A}\ $. Therefore,
again $e^{-\gt \cK_0} f = 0$ for $\gt \ge T$.

Furthermore, if $B(\cdot)$ is a \textit{strongly measurable} family of generators of contraction semigroups
on $X$, i.e. $B(\cdot): \cI \longrightarrow \cG(1,0)$ (see \cite{Kato1978}, Ch.IX, \S 1.4), then the induced
multiplication operator $\cB$ :
\begin{align}\label{eq:1.65}
(\cB f)(t) &:= B(t) f(t) \ ,\\
f \in \dom(\cB) &:= \left\{f \in L^p(\cI,X):
\!\!\!\!\begin{matrix}
& f(t) \in \dom(B(t)) \, \mbox{for a.e.} \; t \in \cI\\
& B(t)f(t) \in L^p(\cI,X)
\end{matrix}\right\}\, ,\nonumber
\end{align}
is a generator of a contraction semigroup on $L^p(\cI,X)$.

In \cite{NeiStephanZag2016} it was assumed that $\{B(t)\}_{t\in \cI}$ is a strongly measurable family of
generators of contraction semigroups and that $A$ is a generator of a bounded holomorphic semigroup with
$0\in \rho(A)$ for simplicity. Moreover, we supposed that the following conditions are satisfied:
\begin{enumerate}

\item[(i)] $\dom(A^\ga) \subseteq \dom(B(t))$ for a.e. $t \in \cI$ and some $\ga \in (0,1)$ such that
\begin{equation*}
\esssup_{t\in \cI}\|B(t)A^{-\ga}\|_{\cB(X)} < \infty\, ;
\end{equation*}

\item[(ii)] $\dom(A^*) \subseteq \dom(B(t)^*)$ for a.e. $t \in \cI$ such that
\begin{equation*}
\esssup_{t\in \cI}\|B(t)^*(A^{-1})^*\|_{\cB(X)} < \infty\, ;
\end{equation*}

\item[(iii)] there is a $\gb \in (\ga,1)$ and $L_\gb > 0$ such that
\be\label{eq:1.19}
\|A^{-1}(B(t) - B(s))A^{-\ga}\|_{\cB(X)} \le L_\gb |t-s|^\gb, \quad t,s \in \cI.
\ee
\end{enumerate}

Under these assumptions it turns out that $\cK := \cK_0 + \cB$ is a generator of a contraction evolution
semigroup, i.e there is a propagator $\{U(t,s)\}_{(t,s)\in \gD}$ such that the representation
(\ref{eq:1.5}) is valid.
Moreover, we prove in \cite{NeiStephanZag2016} the Trotter product formula converges in the operator norm with
convergence rate $O({1}/{n^{\gb-\ga}})$:
\be
\sup_{\gt \ge 0}\left\|\left(e^{-\gt\cK_0/n}e^{-\gt \cB/n}\right)^n -
e^{-\gt \cK}\right\|_{\cB(L^p(\cI,X))} = O({1}/{n^{\gb-\ga}}) \ \nonumber.
\ee

We comment that if $B(\cdot): \cI \longrightarrow \cB(X)$ is a H\"older continuous function with H\"older
exponent $\gb \in (0,1)$, then the assumptions (i)-(iii) are satisfied for any $\ga \in (0,\gb)$.
Then our results \cite{NeiStephanZag2016} yield  that
\begin{equation}\label{eq:1.191}
\sup_{\gt \ge 0}\left\|\left(e^{-\gt\cK_0/n}e^{-\gt \cB/n}\right)^n -
e^{-\gt \cK}\right\|_{\cB(L^p(\cI,X))} = O({1}/{n^{\gga}}) \ ,
\end{equation}
holds for any $\gga \in  (0,\gb)$. Moreover, in this case the perturbation of the
shift semigroup (\ref{eq:1.61}) by a bounded generator (\ref{eq:1.65}) gives an evolution
semigroup with generator $D_0 + \cB$. Then as a corollary of (\ref{eq:1.191}) for $\cA = 0$,
we get the Trotter product estimate
\be\label{eq:1.23}
\sup_{\gt \ge 0}\left\|\left(e^{-\gt D_0/n}e^{-\gt \cB/n}\right)^n -
e^{-\gt (D_0 + \cB)}\right\|_{\cB(L^p(\cI,X))} = O({1}/{n^{\gga}}) \ .
\ee

The aim of our note is to show that the convergence rate \eqref{eq:1.23} is close
to the optimal one. To this end we consider the simple case, when $X = \C$ and we put for simplicity
$\cI := [0,1]$.

The\textit{ main results} of this paper can be summarized as follows:\\
If the operator $\cB$ is equal to the multiplication operator $Q$ induced by a bounded
measurable function $q(\cdot): \cI \longrightarrow \dC$ in $L^p(\cI)$, then one can verify that the
condition \eqref{eq:1.19} is equivalent to $q(\cdot) \in C^{0,\gb}(\cI)$, see definition below.
In this case the convergence
rate is
%If $q(\cdot) \in C^{0,\gb}(\cI)$, then we verify the
%
%
\be\label{eq:1.24}
\sup_{\gt \ge 0}\left\|e^{-\gt (D_0 + Q)}-
\left(e^{-\gt D_0/n}e^{-\gt Q /n}\right)^n \right\|_{\cB(L^p(\cI,X))}
= O(1/n^{\gb}) \ .
\ee
This result remains true if $q(\cdot)$ is Lipschitz continuous, i.e. $\gb = 1$. But if $q(\cdot)$
is \textit{only} continuous, then
\be\label{eq:1.25}
\sup_{\gt \ge 0}\left\|e^{-\gt (D_0 + Q)}-\left(e^{-\gt D_0/n}
e^{-\gt Q/n}\right)^n\right\|_{\cB(L^p(\cI,X))}
= o(1) \ .
\ee
Moreover, for any convergent to zero sequence $\gd_n > 0$, $n \in \dN$, there exists a continuous
function $q(\cdot)$ such that
\be\label{eq:1.26}
\sup_{\gt \ge 0}\left\|e^{-\gt (D_0 + Q)} -
\left(e^{-\gt D_0/n}e^{-\gt Q/n}\right)^n\right\|_{\cB(L^p(\cI,X))} = \go(\gd_n) \ ,
\ee
where the Landau \textit{symbol} $\go(\cdot)$ is defined below.

Finally, there is an example of a bounded measurable function $q(\cdot)$ such that
\be\label{eq:1.27}
\limsup_{n\to\infty} \;\sup_{\gt \ge 0}\left\|e^{-\gt (D_0 + Q)} -
\left(e^{-\gt D_0/n}e^{-\gt Q/n}\right)^{n}\right\|_{\cB(L^p(\cI,X))} > 0 \ .
\ee
Hence, in contrast to the holomorphic case, when the \textit{dominating} operator is a generator of a
holomorphic semigroup (\ref{eq:1.9}), the Trotter product formula \eqref{eq:1.27} with dominating
generator $D_0$, may \textit{not} converge in the operator-norm.

The paper is organized as follows. In Section \ref{TrPrForm} we reformulate the convergence of the
Trotter product formula in terms of the corresponding evolutions semigroups. In Section \ref{BoundPert}
we prove the results \eqref{eq:1.24}-\eqref{eq:1.27}.

We conclude this section by few remarks concerning \textbf{notation} used in this paper.
\begin{enumerate}

\item We use a definition of the \textit{generator} $C$ of a semigroup (\ref{eq:1.8}), which differs from the
standard one by a \textit{minus} \cite{Kato1980}.

\item Furthermore, we widely use the so-called \textit{Landau symbols}:
\begin{align}
g(n) &= O(f(n)) \Longleftrightarrow \limsup_{n\to\infty} \left|\frac{g(n)}{f(n)}\right| < \infty \ ,
\nonumber \\
g(n) &= o(f(n)) \Longleftrightarrow \limsup_{n\to\infty} \left|\frac{g(n)}{f(n)}\right| = 0 \ ,
\nonumber \\
g(n) &= \gT(f(n)) \Longleftrightarrow 0 < \liminf_{n\to\infty} \left|\frac{g(n)}{f(n)}\right|
\le \limsup_{n\to\infty} \left|\frac{g(n)}{f(n)}\right| < \infty \ , \nonumber \\
g(n) &=  \go(f(n)) \Longleftrightarrow \limsup_{n\to\infty} \left|\frac{g(n)}{f(n)}\right| = \infty \ .
\nonumber
\end{align}

\item
We use the notation $C^{0,\beta}(\cl I)=\{f:\cl I\rightarrow \C:
\mathrm{there~is~some~~ } K > 0 \mathrm{~~such~that~ } |f(x)-f(y)|\leq K |x-y|^\beta\}$
for $\beta \in(0,1]$.

\end{enumerate}
%%%%%%%%%%%%%%%%%%%%%%%%%%%%%%%%%%%%%%%%%%%% Section %%%%%%%%%%%%%%%%%%%%%%%%%%%%%%%%%%%%%%%%%%%%%%%%%%%%%%%%
\section{Trotter product formula and evolution semigroups} \label{TrPrForm}
%%%%%%%%%%%%%%%%%%%%%%%%%%%%%%%%%%%%%%%%%%%%%%%%%%%%%%%%%%%%%%%%%%%%%%%%%%%%%%%%%%%%%%%%%%%%%%%%%%%%%%%%%%%%%
Below we consider the Banach space $L^p(\cI,X)$ for  $\cI := [0, T]$, $p \in [1,\infty)$. Recall that
semigroup $\{\cU(\gt)\}_{\gt \ge 0}$, on the Banach space $L^p(\cI,X)$ is called an \textit{evolution}
semigroup if there is a propagator $\{U(t,s)\}_{(t,s)\in\gD}$ such that the representation
(\ref{eq:1.5}) holds.

Let $\cK_0$ be the generator of an evolution semigroup $\{\cU_0(\gt)\}_{\gt \ge 0}$ and let $\cB$
be a multiplication operator induced by a measurable family $\{B(t)\}_{t\in\cI}$ of generators of
contraction semigroups. Note that in this case the multiplication operator $\cB$ (\ref{eq:1.65}) is a
generator of a contraction semigroup $(e^{- \tau \, \cB} f)(t) = e^{- \tau \, B(t)} f(t)$,
on the Banach space $L^p(\cI,X)$. Since $\{\cU_0(\gt)\}_{\gt \ge 0}$ is an evolution semigroup, then by
definition (\ref{eq:1.5})
there is a propagator $\{U_0(t,s)\}_{(t,s) \in \gD}$ such that the representation
\begin{equation*}
(\cU_0(\gt)f)(t) = U_0(t,t-\gt)\chi_\cI(t-\gt)f(t-\gt), \quad f \in L^p(\cI,X),
\end{equation*}
is valid for a.e. $t \in \cI$ and $\gt \ge 0$. Then we define
\begin{equation*}
G_j(t,s;n) := U_0(s + j\tfrac{(t-s)}{n},s+ (j-1)\tfrac{(t-s)}{n})
e^{-\frac{(t-s)}{n} B\big(s + (j-1)\tfrac{(t-s)}{n}\big)}
\end{equation*}
where $j \in \{1,2,\ldots,n\}$, $n \in \dN$, $(t,s) \in \gD$, and we set
\begin{equation*}
V_n(t,s) := \prod^{n\,\leftarrow}_{j=1}G_j(t,s;n), \quad n \in \dN, \quad (t,s) \in \gD,
\end{equation*}
where the product is increasingly ordered in $j$ from the right to the left.
Then a straightforward computation shows that the representation
\begin{equation}\label{eq:2.01}
\left(\left(e^{-\gt \cK_0/n}e^{-\gt \cB/n}\right)^n f\right)(t) =
V_n(t,t-\gt)\chi_\cI(t-\gt)f(t-\gt) \ ,
\end{equation}
$f \in L^p(\cI,X)$, holds for each $\gt \ge 0$ and a.e. $t \in \cI$.
%
%%%%%%%%%%%%%%%%%%%%%%%%%%%%%%%%%%%%%%% Proposition %%%%%%%%%%%%%%%%%%%%%%%%%%%%%%%%%%%%%%%%%%%%%%%%%%%%
\begin{prop}\label{prop:2.1}
Let $\cK$ and $\cK_0$ be generators of evolution semigroups on the Banach space $L^p(\cI,X)$ for some
$p \in [1,\infty)$. Further, let $\{B(t)\in \cG(1,0)\}_{t\in \cI}$ be a strongly measurable family of
generators of contraction on $X$ semigroups. Then
\be\label{eq:2.0}
\sup_{\gt\ge 0}\left\|e^{-\gt \cK} - \left(e^{-\gt \cK_0/n}e^{-\gt \cB/n}\right)^n\right\|_{\cB(L^p(\cI,X))}
 = \esssup_{(t,s)\in \gD}\|U(t,s) - V_n(t,s)\|_{\cB(X)}, \quad n\in \dN.
\ee
\end{prop}
%%%%%%%%%%%%%%%%%%%%%%%%%%%%%%%%%%%%%%%%%%%%%%%%%%%%%%%%%%%%%%%%%%%%%%%%%%%%%%%%%%%%%%%%%%%%%%%%%%%%%%%%%%%
%%%%%%%%%%%%%%%%%%%%%%%%%%%%%%%%%%%%%%%%%%%%%%%%%%%%%%%%%%%%%%%%%%%%%%%%%%%%%%%%%%%%%%%%%%%%%%%%%%%%%%%%%%%
\begin{proof}
Let $\{L(\gt)\}_{\gt \ge 0}$ be the left-shift semigroup on the Banach space $\gotX = L^p(\cI,X)$:
\begin{equation*}
(L(\gt)f)(t) = \chi_\cI(t+\gt)f(t+\gt), \quad f \in L^p(\cI,X).
\end{equation*}
Using that we get
\begin{align*}
\left(L(\gt)\left(e^{-\gt \cK} - \left(e^{-\gt/n \cK_0}e^{-\gt \cB/n}\right)^n\right)f\right)(t) =
\left\{U(t+\gt,t) - V_n(t+\gt,t)\right\}\chi_\cI(t+\gt)f(t) \ ,
\end{align*}
for $\gt \ge 0$ and a.e. $t \in \cI$. It turns out that for each $n \in \dN$ the operator
$L(\gt)\left(e^{-\gt \cK} - \left(e^{-\gt/n \cK_0}e^{-\gt \cB/n}\right)^n\right)$ is a multiplication operator
 induced by $\{(U(t+\gt,t) - V_n(t+\gt,t))\chi_\cI(t+\gt)\}_{t\in\cl I}$. Therefore,
\be
\left\|L(\gt)\left(e^{-\gt \cK} - \left(e^{-\gt\cK_0/n}e^{-\gt \cB/n}\right)^n\right)\right\|_{\cB(\gotX)}=
\esssup_{t\in \cI}\|U(t+\gt,t) - V_n(t+\gt,t)\|_{\cB(X)}\chi_\cI(t+\gt) \ , \nonumber
\ee
for each $\gt \ge 0$. Note that one has
\be
\sup_{\gt \ge 0}\left\|L(\gt)\left(e^{-\gt \cK} -
\left(e^{-\gt\cK_0/n}e^{-\gt \cB/n}\right)^n\right)\right\|_{\cB(\gotX)} = \esssup_{\gt \ge 0}\left\|L(\gt)\left(e^{-\gt \cK} - \left(e^{-\gt\cK_0/n}
e^{-\gt \cB/n}\right)^n\right)\right\|_{\cB(\gotX)} .
\nonumber
\ee
This is based on the fact that if $F(\cdot): \dR_+ \longrightarrow \cB(\gotX)$ is strongly continuous, then
$\sup_{\gt \ge 0}\|F(\gt)\|_{\cB(\gotX)} = \esssup_{\gt \ge 0}\|F(\gt)\|_{\cB(\gotX)}$.
Hence, we find
\be
\sup_{\gt \ge 0}\left\|L(\gt)\left(e^{-\gt \cK} -
\left(e^{-\gt\cK_0/n}e^{-\gt \cB/n}\right)^n\right)\right\|_{\cB(\gotX)}=
\esssup_{\gt \ge 0}\esssup_{t\in \cI}\|U(t+\gt,t) - V_n(t+\gt,t))\|_{\cB(X)}\chi_\cI(t+\gt).\nonumber
\ee
Further, if $\Phi(\cdot,\cdot): \dR_+ \times \cI \longrightarrow \cB(X)$ is a strongly
measurable function, then
\begin{equation*}
\esssup_{(\gt,t)\in\dR_+\times\cI}\|\Phi(\gt,t)\|_{\cB(X)} =
\esssup_{\gt \ge 0}\,\esssup_{t\in\cI}\|\Phi(\gt,t)\|_{\cB(X)} .
\end{equation*}
Then, taking into account two last equalities, one obtains
\begin{align*}
\sup_{\gt \ge 0}\left\|L(\gt)\left(e^{-\gt \cK} -
\left(e^{-\gt\cK_0/n}e^{-\gt \cB/n}\right)^n\right)\right\|_{\cB(\gotX)} &=
\esssup_{(\gt,t)\in \dR_+\times \cI}\|U(t+\gt,t) - V_n(t+\gt,t)\|_{\cB(X)}\chi_\cI(t+\gt)=\nonumber\\
& =
\esssup_{(t,s)\in \gD}\|U(t,s) - V_n(t,s)\|_{\cB(X)} \ , \nonumber
\end{align*}
that proves (\ref{eq:2.0})
\end{proof}

%%%%%%%%%%%%%%%%%%%%%%%%%%%%%%%%%%%%%%%%%%%%% Section %%%%%%%%%%%%%%%%%%%%%%%%%%%%%%%%%%%%%%%%%%%%%%%%%%%%%%%%
\section{Bounded perturbations of the shift semigroup generator}\label{BoundPert}
%%%%%%%%%%%%%%%%%%%%%%%%%%%%%%%%%%%%%%%%%%%%% subsection %%%%%%%%%%%%%%%%%%%%%%%%%%%%%%%%%%%%%%%%%%%%%%%%%%%%%
\subsection{Basic facts}\label{BoundPert1}
%%%%%%%%%%%%%%%%%%%%%%%%%%%%%%%%%%%%%%%%%%%%%%%%%%%%%%%%%%%%%%%%%%%%%%%%%%%%%%%%%%%%%%%%%%%%%%%%%%%%%%%%%%%%%%
We study bounded perturbations of the evolution generator  $D_0$ (\ref{eq:1.6}).
To do this aim we consider $\cl I =[0,1]$, $X= \C$ and we denote by  $L^p(\cl I)$ the Banach
space $L^p(\cl I, \C)$.

For $t \in \cl I$, let $q: t \mapsto q(t) \in L^\infty(\cl I)$.
Then, $q$ induces a bounded multiplication operator $Q$ on the Banach space $L^p(\cl I)$:
\begin{align}
(Qf)(t) = q(t) f(t), ~~f\in L^p(\cl I). \nonumber
\end{align}
For simplicity we assume that $q\geq 0$.
Then $Q$ generates on $L^p(\cl I)$ a contraction semigroup $\{e^{- \tau Q}\}_{\tau \geq 0}$.
Since generator $Q$ is bounded, the closed operator $\cl A:= D_0 + Q$, with domain $\dom(\cl A) = \dom(D_0)$,
is generator of a semigroup on $L^p(\cl I)$. By \cite{Trotter1959}, the Trotter product formula in the strong topology follows immediately
\begin{equation}\label{eq:3.00}
 \left(e^{-\gt D_0/n}e^{-\gt Q/n}\right)^n f \rightarrow e^{-\gt (D_0+Q)}f,  \quad f\in L^p(\cI),
\end{equation}
uniformly in $\tau\in[0,T]$ on bounded time intervals.

Following \cite[\S 5]{Chern1974}, we define on $X= \C$ a family of
bounded operators $\{V(t)\}_{t\in\cl I}$ by
\begin{align*}
 V(t) : =  e^{- \int_0^t ds q(s) } \ .
\end{align*}
Note that for almost every $t\in \cl I$ these operators are positive. Then $V^{-1}(t)$ exists
and it has the form
\begin{align*}
V^{-1}(t) =  e^{ \int_0^t ds q(s) }.
\end{align*}
The operator families $\{V(t)\}_{t\in\cl I}$ and $\{V^{-1}(t)\}_{t\in\cl I}$ induce two
bounded multiplication operators $\cl V $ and $\cl V ^{-1}$ on $L^p(\cl I)$, respectively. Then invertibility implies that
$\cl V \ \cl V^{-1} = \cl V^{-1} \, \cl V = Id|_{L^p}$.
Using the operator $\cl V$ one easily verifies that $D_0+Q$ is {similar} to $D_0$, i.e. one has
\begin{align}
\cl V^{-1}(D_0 + Q)\cl V= D_0, ~~\mathrm{or}~~D_0 + Q= \cl V D_0 \cl V^{-1} \ . \nonumber
\end{align}
Hence, the semigroup generated on $L^p(\cl I)$  by $D_0 + Q$ gets the explicit form:
\begin{align}\label{eq:3.0}
\left(e^{-\tau(D_0 + Q)}f\right)(t) = \left(\cl Ve^{-\tau D_0} \cl V^{-1} f\right)(t) = e^{-\int_{t-\tau}^t q(y) dy}f(t-\tau)\chi_{\cl I}(t-\tau) \ .
\end{align}
Since by (\ref{eq:1.5}) the propagator $U(t,s)$ that corresponds to evolution semigroup (\ref{eq:3.0}) is
defined by
\begin{align}
\left(e^{-\tau(D_0 + Q)}\right)f(t) = U(t, t-\tau) f(t-\tau) \chi_{\cl I}(t-\tau) \ , \nonumber
\end{align}
we deduce that it is equal to $U(t, s) = e^{-\int_s^t dy \, q(y) }$.

Now we study the corresponding Trotter product formula. For a fixed $\tau \geq 0$ and $n\in \mathbb{N}$,
we define approximation $V_n$ by
\begin{equation*}
\left(\left(e^{- \tau D_0/n}e^{- \tau Q/n}  \right)^n f\right)(t) =:
V_n(t,t-\gt)\chi_\cI(t-\gt)f(t-\gt) \ .
\end{equation*}
Then by straightforward calculations, similar to  (\ref{eq:2.01}), one finds that
\begin{equation*}
V_n(t,s)= e^{-\tfrac{t-s}{n} \sum_{k=0}^{n-1} q(s + k\tfrac{t-s}{n})},\quad (t,s) \in \gD \ .
\end{equation*}
%
%%%%%%%%%%%%%%%%%%%%%%%%%%%%%%%%%%%%%%%%%%%%%%%%%%%%% Proposition %%%%%%%%%%%%%%%%%%%%%%%%%%%%%%%%%%%%%%%%%%%%
\begin{prop}\label{prop:3.1}
Let $q \in L^\infty(\cI)$ be non-negative. Then
\be
\sup_{\gt \ge 0}\left\|e^{-\gt(D_0 + Q)} - \left(e^{-\gt D_0/n}e^{-\gt Q/n}\right)^n\right\|_{\cB(L^p(\cl I))}
=
\gT\left(\esssup_{(t,s)\in\gD}\Big|\int^t_s q(y)dy - \frac{t-s}{n} \sum_{k=0}^{n-1} q(s +
k\tfrac{t-s}{n})\Big|\right)
\ee
as $n\to\infty$, where $\gT$ is the Landau symbol defined in Section \ref{Intr}.
\end{prop}
%%%%%%%%%%%%%%%%%%%%%%%%%%%%%%%%%%%%%%%%%%%%%%%%%%%%%%%%%%%%%%%%%%%%%%%%%%%%%%%%%%%%%%%%%%%%%%%%%%%%%%%%%%%%%
%%%%%%%%%%%%%%%%%%%%%%%%%%%%%%%%%%%%%%%%%%%%%%%%%%%%%%%%%%%%%%%%%%%%%%%%%%%%%%%%%%%%%%%%%%%%%%%%%%%%%%%%%%%%%
\begin{proof}
First, by Proposition \ref{prop:2.1} and by $U(t, s) = e^{-\int_s^t dy \, q(y) }$ we obtain
\be\label{eq:3.01}\sup_{\gt \ge 0}\left\|e^{-\gt(D_0 + Q)} -
\left(e^{-\gt D_0/n}e^{-\gt Q/n}\right)^n\right\|_{\cB(L^p(\cl I))}= \esssup_{(t,s)\in \gD}\left|e^{-\int^t_s dy \, q(y)} - e^{-\tfrac{t-s}{n} \sum_{k=0}^{n-1} q(s +
k\tfrac{t-s}{n})}\right| \ .
\ee
Then, using the inequality
\begin{equation*}
e^{-\max\{x,y\}}|x-y| \le |e^{-x} - e^{-y}| \le |x-y|, \quad 0 \le x, y \ ,
\end{equation*}
for $0 \leq s < t \leq 1$ one finds the estimates
\bed
e^{-\|q\|_{L^\infty}} R_n(t,s;q) \le \\
\Big|e^{-\int^t_s dy \, q(y)} - e^{-\tfrac{t-s}{n}
\sum_{k=0}^{n-1} q(s + k\tfrac{t-s}{n})}\Big| \le R_n(t,s;q) \ ,
\eed
where
\begin{equation}\label{eq:3.02}
R_n(t,s,q) := \Big|\int^t_s dy \, q(y) - \frac{t-s}{n} \sum_{k=0}^{n-1} q(s + k\tfrac{t-s}{n})\Big| \ ,
\quad (t,s) \in \gD \ .
\end{equation}
Hence, for the left-hand side of (\ref{eq:3.01}) we get the estimate
\begin{equation*}
e^{-\|q\|_{L^\infty}} R_n(q) \le \sup_{\gt \ge 0}\left\|e^{-\gt(D_0 + Q)} - \left(e^{-\gt D_0/n}
e^{-\gt Q/n}\right)^n\right\|_{\cB(L^p)} \le R_n(q) \ ,
\end{equation*}
where $R_n(q) := \esssup_{(t,s)\in \gD}R_n(t,s;q)$, $n \in \dN$. These estimates together with definition
of $\gT$ prove the assertion.
\end{proof}
%%%%%%%%%%%%%%%%%%%%%%%%%%%%%%%%%%%%%%%%%%%%%%%%%%%%%%%%%%%%%%%%%%%%%%%%%%%%%%%%%%%%%%%%%%%%%%%%%%%%%%%%%%%%%%
Note that by virtue of (\ref{eq:3.02}) and Proposition \ref{prop:3.1} the operator-norm
convergence rate of the Trotter product formula for the pair $\{D_0 , Q\}$ coincides with the convergence
rate of the integral Darboux-Riemann sum approximation of the Lebesgue integral.

%%%%%%%%%%%%%%%%%%%%%%%%%%%%%%%%%%%%%%%%%%%%% subsection %%%%%%%%%%%%%%%%%%%%%%%%%%%%%%%%%%%%%%%%%%%%%%%%%%%%%
\subsection{Examples}\label{Example}
%%%%%%%%%%%%%%%%%%%%%%%%%%%%%%%%%%%%%%%%%%%%%%%%%%%%%%%%%%%%%%%%%%%%%%%%%%%%%%%%%%%%%%%%%%%%%%%%%%%%%%%%%%%%%%
First we consider the case of a real H\"older-continuous function $q\in C^{0,\gb}(\cI)$.
%%%%%%%%%%%%%%%%%%%%%%%%%%%%%%%%%%%%%%%%%%%%%%% Theorem %%%%%%%%%%%%%%%%%%%%%%%%%%%%%%%%%%%%%%%%%%%%%%%%%%%%%%
\begin{thm}
If {{$q \in C^{0,\gb}(\cI)$} is non-negative}, then
\begin{equation*}
\sup_{\gt \ge 0}\left\|e^{-\gt(D_0+Q)} - \left(e^{-\gt D_0/n}e^{-\gt Q/n}\right)^n\right\| =
O({1}/{n^\gb}) \ ,
\end{equation*}
as $n \to \infty$.
\end{thm}
%%%%%%%%%%%%%%%%%%%%%%%%%%%%%%%%%%%%%%%%%%%%%%%%%%%%%%%%%%%%%%%%%%%%%%%%%%%%%%%%%%%%%%%%%%%%%%%%%%%%%%%%%%%%%%
%%%%%%%%%%%%%%%%%%%%%%%%%%%%%%%%%%%%%%%%%%%%%%%%%%%%%%%%%%%%%%%%%%%%%%%%%%%%%%%%%%%%%%%%%%%%%%%%%%%%%%%%%%%%%%
\begin{proof}
One has
\begin{equation*}
\int^t_s dy \, q(y) - \frac{t-s}{n}\sum^{n-1}_k q(s + \tfrac{k}{n}(t-s))
= \sum^{n-1}_{k=0}\int^{\tfrac{k+1}{n}(t-s)}_{\tfrac{k}{n}(t-s)}dy \left(q(s+y) -
q(s+\tfrac{k}{n}(t-s))\right)\ ,
\end{equation*}
which yields the estimate
\begin{equation*}
\Big|\int^t_s dy \, q(y) - \frac{t-s}{n}\sum^{n-1}_k q(s + \tfrac{k}{n}(t-s))\Big|\\
\le \sum^{n-1}_{k=0}\int^{\tfrac{k+1}{n}(t-s)}_{\tfrac{k}{n}(t-s)}
dy \left|q(s+y) - q(s+\tfrac{k}{n}(t-s))\right| \ .
\end{equation*}
Since $q \in C^{0,\gb}(\cI)$,  there is a constant $L_\gb > 0$ such that for $y\in[\frac k n (t-s),
\frac{k+1} n (t-s)]$ one has
\begin{equation*}
\left|q(s+y) - q(s+\tfrac{k}{n}(t-s)\right| \le L_\gb|y-\tfrac{k}{n}(t-s)|^\gb \le
L_\gb \frac{(t-s)^\gb}{n^\gb}\ .
\end{equation*}
Hence, we find
\begin{equation*}
\Big|\int^t_s q(y)dy - \frac{t-s}{n}\sum^{n-1}_k q(s + \tfrac{k}{n}(t-s))\Big|
\le L_\gb \frac{(t-s)^{1+\gb}}{n^\gb} \le L_\gb \frac{1}{n^\gb} \ ,
\end{equation*}
which proves
\begin{equation*}
\esssup_{(t,s)\in \gD}\Big|\int^t_s q(y)dy - \frac{t-s}{n}\sum^{n-1}_k q(s + \tfrac{k}{n}(t-s))\Big|
= O\left(\frac{1}{n^\gb}\right)\, .
\end{equation*}
Applying now Proposition \ref{prop:3.1} one completes the proof.
\end{proof}
%%%%%%%%%%%%%%%%%%%%%%%%%%%%%%%%%%%%%%%%%%%%%%%%%%%%%%%%%%%%%%%%%%%%%%%%%%%%%%%%%%%%%%%%%%%%%%%%%%%%%%%%%%%%%%%

It is a natural question: what happens, when $q$ is only continuous?
%
%%%%%%%%%%%%%%%%%%%%%%%%%%%%%%%%%%%%%%%%%%%%% Theorem %%%%%%%%%%%%%%%%%%%%%%%%%%%%%%%%%%%%%%%%%%%%%%%%%%%%%%%%%
\begin{thm}
If $q: \cl I \rightarrow \C$ is continuous and non-negative, then
\begin{equation}\label{eq:4.10}
\left\|e^{-\gt(D_0+Q)} - \left(e^{-\gt D_0/n}e^{-\gt Q/n}\right)^n\right\| = o(1) \ ,
\end{equation}
as $n\to \infty$.
\end{thm}
%%%%%%%%%%%%%%%%%%%%%%%%%%%%%%%%%%%%%%%%%%%%%%%%%%%%%%%%%%%%%%%%%%%%%%%%%%%%%%%%%%%%%%%%%%%%%%%%%%%%%%%%%%%%%%%
%%%%%%%%%%%%%%%%%%%%%%%%%%%%%%%%%%%%%%%%%%%%%%%%%%%%%%%%%%%%%%%%%%%%%%%%%%%%%%%%%%%%%%%%%%%%%%%%%%%%%%%%%%%%%%%
\begin{proof}
Since $q(\cdot)$ is continuous, then for any $\varepsilon > 0$ there is $\gd > 0$ such that for
$|y-x| < \gd$ we have $|q(y) - q(x)| < \varepsilon$, $y,x \in \cI$.
Therefore, if $1/n < \gd$, then for $y \in (\frac{k}{n}(t-s),\frac{k+1}{n}(t-s))$ we have
\begin{equation*}
|q(s+y) - q(s+\tfrac{k}{n}(t-s))| < \varepsilon, \quad (t,s) \in \gD \ .
\end{equation*}
Hence,
\begin{equation*}
\Big|\int^t_s q(y)dy - \frac{t-s}{n}\sum^{n-1}_k q(s + \tfrac{k}{n}(t-s))\Big|
\le \varepsilon (t-s) \le \varepsilon \ ,
\end{equation*}
which yields
\begin{equation*}
\esssup_{(t,s)\in\gD}\Big|\int^t_s q(y)dy - \frac{t-s}{n}\sum^{n-1}_k q(s + \tfrac{k}{n}(t-s))\Big| = o(1) \ .
\end{equation*}
Now it remains only to apply Proposition \ref{prop:3.1}.
\end{proof}
%%%%%%%%%%%%%%%%%%%%%%%%%%%%%%%%%%%%%%%%%%%%%%%%%%%%%%%%%%%%%%%%%%%%%%%%%%%%%%%%%%%%%%%%%%%%%%%%%%%%%%%%%%%%%%%

We comment that for a general continuous $q$ one can say nothing about the convergence rate. Indeed, it can
be shown that in (\ref{eq:4.10}) the convergence to zero can be arbitrary slow.
%
%%%%%%%%%%%%%%%%%%%%%%%%%%%%%%%%%%%%%%%%%%%%%%%% Theorem %%%%%%%%%%%%%%%%%%%%%%%%%%%%%%%%%%%%%%%%%%%%%%%%%%%%%%
\begin{thm}\label{thm: 4.3}
 Let $\gd_n>0$ be a sequence with $\gd_n \to 0$ as $n \to \infty$. Then there exists a continuous
 function $q:\cI = [0,1] \rightarrow \dR$ such that
 \be\label{eq:4.11}
\sup_{\gt \ge 0}\left\|e^{-\gt(D_0 + Q)} - \left(e^{-\gt D_0/n}
e^{-\gt Q/n}\right)^n\right\|_{\cl B( L^p(\cl I))} = \go(\gd_n)
 \ee
as $n\to\infty$, where $\omega$ is the Landau symbol defined in Section \ref{Intr}.
\end{thm}
%%%%%%%%%%%%%%%%%%%%%%%%%%%%%%%%%%%%%%%%%%%%%%%%%%%%%%%%%%%%%%%%%%%%%%%%%%%%%%%%%%%%%%%%%%%%%%%%%%%%%%%%%%%
\begin{proof}
Taking into account Theorem 6 of \cite{WalsSewell1937}, we find that for any sequence
$\{\gd_n\}_{n\in\dN}$, $\gd_n > 0$ satisfying $\lim_{n\to\infty}\gd_n = 0$
there exists a continuous function $f(\cdot): [0,2\pi] \longrightarrow \dR$ such that
\begin{equation*}
\left|\int^{2\pi}_0  f(x)\, dx - \frac{2\pi}{n}\sum^n_{k=1}f(2k\pi/n)\right| = \go(\gd_n) \ ,
\end{equation*}
as $n \to \infty$. Setting $q(y) := f(2\pi(1-y))$, $y \in [0,1]$, we get a continuous function
$q(\cdot): [0,1] \longrightarrow \dR$,
such that
\begin{equation*}
\left|\int^{1}_0 q(y)dy- \frac{1}{n}\sum^{n-1}_{k=0}q(k/n)\right| = \go(\gd_n) \ .
\end{equation*}
Because $q(\cdot)$ is continuous we find
\begin{equation*}
\esssup_{(t,s)\in \gD}\Big|\int^t_s q(y) \, dy - \frac{t-s}{n}\sum^{n-1}_{n=0}q(s + k\tfrac{t-s}{n})\Big|\\
\ge \Big|\int^{1}_0 q(y)\,dy- \frac{1}{n}\sum^{n-1}_{k=0}q(k/n)\Big| \ ,
\end{equation*}
which yields
\begin{equation*}
\esssup_{(t,s)\in \gD}\Big|\int^t_s q(y)\,dy - \frac{t-s}{n}\sum^{n-1}_{n=0}q(s + k\tfrac{t-s}{n})\Big| =
\go(\gd_n) \ .
\end{equation*}
Applying now Proposition \ref{prop:3.1} we prove \eqref{eq:4.11}.
\end{proof}
%%%%%%%%%%%%%%%%%%%%%%%%%%%%%%%%%%%%%%%%%%%%%%%%%%%%%%%%%%%%%%%%%%%%%%%%%%%%%%%%%%%%%%%%%%%%%%%%%%%%%%%%%%%%%%%

Our final comment concerns the case when $q$ is only \textit{measurable}. Then it can happen that the
Trotter product formula for that pair $\{D_0 , Q\}$ does not converge in the operator-norm topology.

%%%%%%%%%%%%%%%%%%%%%%%%%%%%%%%%%%%%%%%%%%%%%%%%%% Theorem %%%%%%%%%%%%%%%%%%%%%%%%%%%%%%%%%%%%%%%%%%%%%%%%%%%%
\begin{thm}\label{thm:4.4}
 There is a non-negative function $q \in L^\infty([0,1])$ such that
 \be\label{eq:4.16}
 \limsup_{n\to\infty}\; \sup_{\gt \ge 0}\left\|e^{-\gt(D_0 + Q)} - \left(e^{-\gt D_0/n}
 e^{-\gt Q/n}\right)^n\right\|_{\cl B( L^p(\cl I))} > 0 \ .
 \ee
\end{thm}
%%%%%%%%%%%%%%%%%%%%%%%%%%%%%%%%%%%%%%%%%%%%%%%%%%%%%%%%%%%%%%%%%%%%%%%%%%%%%%%%%%%%%%%%%%%%%%%%%%%%%%%%%%%%%%%
%%%%%%%%%%%%%%%%%%%%%%%%%%%%%%%%%%%%%%%%%%%%%%%%%%%%%%%%%%%%%%%%%%%%%%%%%%%%%%%%%%%%%%%%%%%%%%%%%%%%%%%%%%%%%%%
\begin{proof}
Let us introduce the open intervals
\begin{equation*}
\begin{split}
 \gD_{0,n} &:= (0,\tfrac{1}{2^{2n+2}}),\\
 \gD_{k,n} &:= (t_{k,n} - \tfrac{1}{2^{2n +2}},t_{k,n} + \tfrac{1}{2^{2n +2}}), \quad k = 1,2,\ldots,2^n-1,\\
 \gD_{2^n,n}&:= (1-\tfrac{1}{2^{2n+2}},1),
\end{split}
\end{equation*}
$n \in \dN$, where
\begin{equation*}
t_{k,n} = \frac{k}{2^n}, \quad k = 0,\ldots, n,\quad n \in \dN.
\end{equation*}
Notice that $t_{0,n} = 0$ and $t_{2^n,n} = 1$. One easily checks that the intervals
$\gD_{k,n}$, $k = 0,\ldots,2^n$, are mutually disjoint. We introduce the open sets
\begin{equation*}
\cO_n = \bigcup^{2^n}_{k = 0} \gD_{k,n} \subseteq \cI, \quad n \in \dN.
\end{equation*}
and
\begin{equation*}
\cO = \bigcup_{n\in\dN}\cO_n \subseteq \cI.
\end{equation*}
Then it is clear that
\begin{equation*}
|\cO_n| = \frac{1}{2^{n+1}}, \quad n \in \dN,
\quad \mbox{and} \quad
|\cO| \le \frac{1}{2}.
\end{equation*}
Therefore, the Lebesgue measure of the closed set $\cC := \cI\setminus\cO \subseteq \cI$ can be
estimated by
\begin{equation*}
|\cC| \ge \frac{1}{2} \ .
\end{equation*}
Using the characteristic function $\chi_{\cC}(\cdot)$ of the set $\cC$ we define
\begin{equation*}
q(t) :=\chi_{\cC}(t), \quad t \in \cI \ .
\end{equation*}
The function $q(\cdot)$ is measurable and it satisfies $0 \le q(t) \le 1$, $t \in \cI$.

Let $\varepsilon \in (0,1)$. We choose $s \in (0,\varepsilon)$ and $t \in (1-\varepsilon,1)$ and  we set
\begin{equation*}
\xi_{k,n}(t,s) := s + k\,\frac{t-s}{2^n}, \quad k = 0,\ldots,2^n-1, \quad n \in \dN, \quad (t,s) \in \gD.
\end{equation*}
Note that $\xi_{k,n}(t,s) \in (0,1)$, $k = 0,\ldots,2^n-1$, $n \in \dN$. Moreover, we have
\begin{equation*}
t_{k,n} - \xi_{k,n}(t,s) =  k\frac{1}{2^n} -s - k\,\frac{t-s}{2^n} = k\frac{1-t+s}{2^n} -s \ ,
\end{equation*}
which leads to the estimate
\begin{equation*}
|t_{k,n} - \xi_{k,n}(t,s)| \le \varepsilon(\frac{k}{2^{n-1}} + 1), \quad k =0,\ldots,2^n-1, \quad n \in \dN \ .
\end{equation*}
Hence
\begin{equation*}
|t_{k,n} - \xi_{k,n}(t,s)| \le 3\varepsilon, \quad k =0,\ldots,2^n-1, \quad n \in \dN.
\end{equation*}
Let $\varepsilon_n := {1}/{(3\cdot 2^{2n+2})}$ for $n \in \dN$. Then we get that $\xi_{k,n}(t,s) \in \gD_{k,n}$
for $k = 0,\ldots,2^n-1$, $n \in \dN$, $s \in (0,\varepsilon_n)$ and for $t \in (1-\varepsilon_n,1)$.

Now let
\begin{equation*}
S_n(t,s;q) := \frac{t-s}{n}\sum^{n-1}_{k=0}q(s + k\tfrac{t-s}{n}), \quad n \in \dN, \quad (t,s) \in \gD \ .
\end{equation*}
We consider
\begin{equation*}
S_{2^n}(t,s;q) = \frac{t-s}{n}\sum^{2^n-1}_{k=0}q(s + k\tfrac{t-s}{2^n}) =
\frac{t-s}{n}\sum^{2^n-1}_{k=0}q(\xi_{k,n}(t,s)),
\end{equation*}
$n \in \dN$, $(t,s) \in \gD$. If $s \in (0,\varepsilon_n)$ and $t \in (1-\varepsilon_n,1)$,
then $S_{2^n}(t,s;q) = 0$, $n \in \dN$ and
\begin{equation*}
\left|\int^t_s q(y)\, dy - S_{2^n}(t,s;q)\right| = \int^t_s q(y) dy, \quad n \in \dN,
\end{equation*}
for $s \in (0,\varepsilon_n)$ and $t \in (1-\varepsilon_n,1)$. In particular, this yields
\bed
\esssup_{(t,s)\in \gD}\left|\int^t_s q(y) dy - S_{2^n}(t,s;q)\right| \ge \esssup_{(t,s)
\in\gD}\int^t_s q(y) dy \ge \int_{\cI} \chi_\cC(y)dy \ge \frac{1}{2} \ .
\eed
Hence, we obtain
\begin{equation*}
\limsup_{n\to\infty}\;\esssup_{(t,s)\in \gD}\left|\int^t_s q(y) dy - S_{2^n}(t,s;q)\right|  \ge \frac{1}{2},
\end{equation*}
and applying Proposition  \ref{prop:3.1} we finish the prove of \eqref{eq:4.16}.
\end{proof}
%%%%%%%%%%%%%%%%%%%%%%%%%%%%%%%%%%%%%%%%%%%%%%%%%%%%%%%%%%%%%%%%%%%%%%%%%%%%%%%%%%%%%%%%%%%%%%%%%%%%%%%%%%%%%%

We note that Theorem \ref{thm:4.4} does not exclude the convergence of the Trotter product formula
for the pair $\{D_0 , Q\}$ in the \textit{strong} operator topology. Examples of this dichotomy are known
for the Trotter-Kato product formula in Hilbert spaces \cite{ITTZ2001}. By virtue of (\ref{eq:3.00}) 
and (\ref{eq:4.16}), Theorem \ref{thm:4.4} yields an example of this dichotomy in Banach spaces.
%%%%%%%%%%%%%%%%%%%%%%%%%%%%%%%%%%%%%%%%%%%%%%%%%%%%%%%%%%%%%%%%%%%%%%%%%%%%%%%%%%%%%%%%%%%%%%%%%%%%%%%%%%%%%%
%%%%%%%%%%%%%%%%%%%%%%%%%%%%%%%%%%%%%%%%%%%%% Acknowledgments %%%%%%%%%%%%%%%%%%%%%%%%%%%%%%%%%%%%%%%%%%%%%%%%
\subsection*{Acknowledgments}
The preparation of the paper was supported by the European Research Council via ERC-2010-AdG no 267802
(``Analysis of Multiscale Systems Driven by Functionals'').
V.A.Z. thanks WIAS for hospitality.

%\bibliographystyle{plain}
%\bibliography{trotter2016}
%%%%%%%%%%%%%%%%%%%%%%%%%%%%%%%%%%%%%%%%%%%%%%%% References %%%%%%%%%%%%%%%%%%%%%%%%%%%%%%%%%%%%%%%%%%%%%%%%%
\def\cprime{$'$}

\end{document}